\documentclass[10pt,journal,letterpaper,twocolumn,twoside,nofonttune]{IEEEtranTCOM}
\normalsize

\usepackage[T1]{fontenc}
\usepackage{amsmath,amssymb,amsfonts}
\usepackage{mathrsfs}
\usepackage{amsbsy}
\usepackage{graphicx}
\usepackage{graphics}
\usepackage{epstopdf}
\usepackage{algorithm}
\usepackage{algorithmic}
\usepackage{subfigure}
\usepackage{cite}
\usepackage{slashbox}
\usepackage{multirow}
\usepackage{ctable}
\usepackage{tikz,pgfplots}

\newtheorem{theorem}{{Theorem}}
\newtheorem{lemma}[theorem]{{Lemma}}

\newcommand{\cC}{{\cal C}}

\newcommand{\cE}{{\cal E}}

\newcommand{\cG}{{\cal G}}

\DeclareMathAlphabet{\mathbfsl}{OT1}{ppl}{b}{it} 

\newcommand{\bU}{\mathbfsl{U}}

\newcommand{\be}[1]{\begin{equation}\label{#1}}
\newcommand{\ee}{\end{equation}}
\newcommand{\eq}[1]{(\ref{#1})}

\renewcommand{\leq}{\leqslant}

\renewcommand{\geq}{\geqslant}

\newcommand{\script}[1]{{\mathscr #1}}

\newcommand{\Tref}[1]{Theo\-rem\,\ref{#1}}

\newcommand{\Lref}[1]{Lem\-ma\,\ref{#1}}
\newcommand{\Cref}[1]{Co\-ro\-lla\-ry\,\ref{#1}}

\newcommand{\deff}{\mbox{$\stackrel{\rm def}{=}$}}
\newcommand{\Strut}[2]{\rule[-#2]{0cm}{#1}}

\newcommand{\hU}{\widehat{\bU}}

\newcommand{\sX}{\script{X}}
\newcommand{\sY}{\script{Y}}

\newcommand{\shalf}{\mbox{\raisebox{.8mm}{\footnotesize $\scriptstyle 1$}
\footnotesize$\!\!\! / \!\!\!$ \raisebox{-.8mm}{\footnotesize
$\scriptstyle 2$}}}

\newcommand{\eps}{{\epsilon}}

\newcommand{\F}{\mathbb{F}}

\newcommand{\Gs}{G^{\otimes s}}

\newcommand{\hatu}{\hat{u}}

\begin{document}

\title{Performance Limits and Practical Decoding of Interleaved Reed-Solomon Polar Concatenated Codes}
%
%
%

\author{%
\authorblockN{\large{
Hessam Mahdavifar, 
        Mostafa El-Khamy, 
        Jungwon Lee,
        Inyup Kang}}\\
\thanks{%
   Hessam Mahdavifar, Mostafa El-Khamy, Junwon Lee and Inyup Kang are with the
   Mobile Solutions Lab, Samsung Research America, San Diego, CA 92121, U.S.A.
   (e-mail: \{h.mahdavifar,\,mostafa.e,\,jungwon2.lee,\,inyup.kang\}@samsung.com).
}}


%
%

\maketitle

\begin{abstract}
A scheme for concatenating the recently invented polar codes with non-binary MDS codes, as Reed-Solomon codes, is considered. By concatenating binary polar codes with interleaved Reed-Solomon codes, we prove that the proposed concatenation scheme captures the capacity-achieving property of polar codes, while having a significantly better error-decay rate. We show that for any $\epsilon > 0$, and total frame length $N$,  the parameters of the scheme can be set such that the frame error probability is less than $2^{-N^{1-\epsilon}}$, while the scheme is still capacity achieving. This improves upon $2^{-N^{0.5-\eps}}$, the frame error probability of Arikan's polar codes. The proposed concatenated polar codes and Arikan's polar codes are also compared for transmission over channels with erasure bursts. We provide a sufficient condition on the length of erasure burst which guarantees failure of the polar decoder. On the other hand, it is shown that the parameters of the concatenated polar code can be set in such a way that the capacity-achieving properties of polar codes are preserved. We also propose decoding algorithms for concatenated polar codes, which significantly improve the error-rate performance at finite block lengths while preserving the low decoding complexity.
\end{abstract}


%

\section{Introduction}
\label{sec:Introduction}
Polar codes, introduced by Arikan \cite{Arikan, AT}, are the most recent breakthrough in coding theory. Polar codes are the first and, currently, the only family of codes with explicit construction (no ensemble to pick from) to achieve the capacity of a certain family of channels (binary input symmetric discrete memory-less channels) as the block length goes to infinity. They have encoding and decoding algorithms with very low complexity. Their encoding complexity is $n \text{log} n$ and their successive cancellation (SC) decoding complexity is $O(n \text{log} n)$, where $n$ is the length of the code. However, at moderate block lengths, their performance does not compete with world's best known codes, which prevents them from being implemented in practice. Also, their error exponent decreases slowly as the block length increases, where the error-decay rate of polar codes under successive cancellation decoding is asymptotically $O(2^{-n^{0.5-\eps}})$. In this paper, we aim at providing techniques to make polar codes more practical, by providing schemes that improve their finite length performance, while preserving their low decoding complexity.

Concatenating inner polar codes with outer linear codes (or other variations of concatenation like parallel concatenation) is a promising path towards making them more practical \cite{BJE,EP}. By carefully constructing such codes, the concatenated construction can inherit the low encoding and decoding complexities of the inner polar code, while having significantly improved error-rate performance, in comparison with the inner polar code. The performance and decoding complexity of the concatenated code will also depend on the outer code used, the concatenation scheme and the decoding algorithms used for decoding the component codes. We chose Reed-Solomon (RS) codes  as outer codes as they are maximal distance separable (MDS) codes, and hence have the largest bounded-distance error-correction capability at a specified code rate. RS codes also have excellent burst error-correction capability. 

Recent investigations have shown the possibility of improving the bound on the error-decay rate of polar codes by concatenating them with RS codes \cite{BJE}. However, this work assumed a conventional method of concatenation, which required the cardinality of the outer RS code alphabet to be exponential in the block length of the inner polar code, which makes it infeasible for implementation in practical systems.  In this paper, we propose a scheme for improving the error-decay rate of polar codes by concatenating them with interleaved block codes. When deploying our proposed scheme with outer interleaved RS block codes, the RS alphabet cardinality is no longer exponential and, in fact, is a design parameter which can be chosen arbitrarily. Furthermore, we show that the code parameters can be set such that the total scheme still achieves the capacity while the error-decay rate is asymptotically $2^{-N^{1-\eps}}$ for any $\eps > 0$, where $N$ is the total block length of the concatenated scheme. This bound provides considerable improvements upon $2^{-N^{0.5 - \eps}}$, the error-decay rate of Arikan's polar codes, and upon the bound of \cite{BJE}. Notice that $2^{-N}$ is the information theoretic upper bound on the probability of error of any capacity achieving code. This can be derived from the main result of \cite{PPV}. Therefore, our concatenated code fills the gap with the ultimate bound on the error decay rate of any capacity achieving code.

Considering transmission over channels with erasure bursts we show that the proposed concatenated polar codes perform well while we prove that Arikan's polar codes are very weak in this regard. For theoretical analysis, transmission over an arbitrary binary symmetric channel channel is considered while the received symbols encounter a single erasure burst. We prove that if the length of erasure burst is at least $2\sqrt{n}-1$, where $n$ is the length of polar code, then the successive cancellation decoder fails to recover the transmitted message. On the other hand, we show that the parameters of the concatenated polar code can be set in such a way that the concatenated polar code approaches the capacity of the channel with the same error-decay rate as in polar codes. For simulations, Gilbert-Elliot model is considered wherein a channel with two states is assumed. The good channel state is a binary erasure channel and the bad channel state is always erasure. For some specific parameters, we show that the concatenated polar code perform well as expected, while the non-concatenated polar code fails with a very high probability of error.

To construct the concatenated polar code at finite block length, we propose a rate-adaptive method to minimize the rate-loss resulting from the outer block code. It is known from the theory of polar codes that not all of the selected good bit-channels have the same performance. Some of the information bits observe very strong and almost noiseless channels, while some other information bits observe weaker channels. This implies that an unequal protection by the outer code is needed, i.e. the  strongest information bit-channels do not need another level of protection by the outer code, while the rest are protected by certain codes whose rates are determined by the error probability of the corresponding bit-channels. Hence, a criterion is established for determining the proper rates of the outer interleaved block codes.

We propose a successive method for decoding the RS-polar concatenated scheme, which is possible by the proposed interleaved concatenation scheme, where the symbols of each RS codeword are distributed over the same coordinates of multiple polar codewords. In the successive cancellation (SC) decoding of the inner polar codes, the very first bits of each polar codeword that are protected by the first outer RS code are decoded first. Then these bits are passed as symbols to the first outer decoder. RS decoding is done on the first RS word to correct any residual errors from  the polar decoders, and pass the updated information back to the SC polar decoders. Then the SC decoders of all polar codes update their first decoded bits, and use that updated information to continue successive decoding for the following bits corresponding to the symbols of the subsequent outer words.  Therefore, the errors from SC decoding do not propagate through the whole polar codeword, which  significantly improves the performance of our scheme. Another main advantage of this proposed scheme, is that all polar codes are decoded in parallel which significantly reduces the decoding latency.

Depending on the chosen outer code and its chosen decoding algorithm, the information exchanged between the inner SC decoders and the outer decoder can be soft information, as log-likelihood ratios (LLRs), or hard decisioned bits. We take advantage of the soft information generated by the successive cancellation decoder of polar codes to perform generalized minimum distance (GMD) list decoding \cite{F,K} for the outer RS code, which enhances the error performance. For further improvements, the SC decoder is modified so that it generates the likelihoods of all the possible RS symbols for GMD decoding. After GMD decoding of each component RS word, the most likely candidate codeword relative to the received word is picked from the list, and SC decoder utilizes the updated RS soft and hard outputs in further decoding of the component polar codes. In the case that outer codes are RS codes, more complex and better soft decoding algorithms exist, e.g. \cite{KV,EM}, however GMD was chosen to preserve the bound on the decoding complexity of the polar codes.

The rest of this paper is organized as follows: In Section\,\ref{sec:two}, some necessary background on Arikan's polar codes is provided. In Section\,\ref{sec:three}, the proposed scheme is explained in more details. We also explain how to modify the scheme to get a rate-adaptive construction, which significantly improves the rate of finite length constructions. In Section\,\ref{sec:four}, we describe our proposed decoding algorithms for the concatenated polar code scheme that improve the performance at finite block length. Asymptotic analysis of error correction performance along with the proof of the bounds on the error-decay rate are provided in Section\,\ref{sec:five}. Also, the performance of the proposed concatenated polar code over channels with erasure burst is discussed. Simulated results are shown in Section \ref{sec:six}. Finally, we conclude the paper in Section\,\ref{sec:seven}.

                                                                       %
                                                                       %
\section{Arikan's polar codes}
\label{sec:two}

A brief overview of the groundbreaking work of
Arikan~\cite{Arikan} 
and others~\cite{AT,Korada,KSU} on polar 
codes and channel polarization is provided in this section. 

The construction of polar codes is based on a phenomenon called \emph{channel polarization} discovered by Arikan \cite{Arikan}. Let 
\be{G-def}
G 
\ = \
\left[ 
\begin{array}{c@{\hspace{1.25ex}}c}
1 & 0\\
1 & 1\\ 
\end{array}
\right]
\ee
The Kronecker powers of $G$ are defined by induction. $G^{\otimes 1} = G$ and for any $i > 1$:
$$
G^{\otimes (i)}
\ = \
\left[ 
\begin{array}{c@{\hspace{1.25ex}}c}
G^{\otimes (i-1)} & 0\\
G^{\otimes (i-1)} & G^{\otimes (i-1)}\\ 
\end{array}
\right]
$$
Observe that $G^{\otimes (i)}$ is a $2^i \times 2^i$ matrix. Let $n = 2^s$. Then the $n \times n$ polarization transform matrix is defined as
$G_n \,\smash{\raisebox{-0.35ex}{\deff}}\, R_n \Gs$, 
where $R_n$ is the bit-reversal permutation matrix defined 
in \cite[Section\,VII-B]{Arikan}. Let $(U_1, U_2,\dots,U_n)$, denoted by $U_1^n$, be a block of $n$ independent and uniform binary random variables. Let also $X_1^n = U_1^n G_n$. $X_i$'s are transmitted through $n$ independent copies of a binary-input discrete memoryless channel (B-DMC) $W$. The output is denoted by $Y_1^n$. This transformation with input $U_1^n$ and output $Y_1^n$ is called the polar transformation. In this transformation, $n$ independent uses of $W$ is split into $n$ bit-channels assuming that a successive cancellation decoder is deployed at the output. Under this decoding method, all the bits $U_1^{i-1}$ are already decoded and are available at the time that $U_i$ is being decoded. This channel is called the $i$-th bit channel and is denoted by $W^{(i)}_n$. The channel polarization theorem  
states that as $n$ goes to infinity, the bit-channels start polarizing meaning that they either become a noise-less channel or a pure-noise channel. The definition of bit-channels and the channel polarization theorem are discussed more precisely next. 

For any discrete memory-less channel $W : \sX \rightarrow \sY$, let $W(y|x)$ denote the probability of receiving $y \in \!\sY$ given that $x \in \!\sX$ was sent, for any $x \in \!\sX$ and $y \in \!\sY$. Let $W^n: \sX^n \rightarrow \sY^n$ denote the channel that results from $n$ independent copies of $W$ in the polar transformation i.e. 
\be{Wn}
W^n\kern-0.5pt(y^n_1|x^n_1) 
\,\ \deff\,\
\prod_{i=1}^n W(y_i|x_i)
\vspace{-0.25ex}
\ee
Then the combined channel $\widetilde{W}$ is defined
with transition probabilities given by
\be{Wtilde}
\widetilde{W}(y^n_1|u^n_1) 
\,\ \deff\,\
W^n\kern-1pt\bigl(y^n_1\hspace{1pt}{\bigm|}\hspace{1pt}u^n_1\hspace{1pt} G_n\bigr)
\kern1pt = \kern2pt
W^n\kern-1pt
\bigl(y^n_1\hspace{1pt}{\bigm|}\hspace{1pt}u^n_1\hspace{1pt} R_n \Gs \bigr)
\ee
This is the channel that the random vector
$(U_1,U_2,\dots,U_n)$ observes through the polar transformation. The transition probabilities for the bit-channel $W^{(i)}_n$ is given as follows:
\be{Wi-def}
W^{(i)}_n\bigl( y^n_1,u^{i-1}_1 | \hspace{1pt}u_i)
\,\ \deff \,\
\frac{1}{2^{n-1}}\hspace{-5pt}
\sum_{u_{i+1}^n \in \{0,1\}^{n-i}} \hspace{-12pt}
\widetilde{W}\Bigl(y^n_1\hspace{1pt}{\bigm|}\hspace{1pt}
u_1^n \Bigr)
\ee

For any B-DMC $W$, the \emph{Bhattacharyya parameter} of $W$ is
$$ 
Z(W)
\,\ \deff\kern1pt
\sum_{y\in\sY} \!\sqrt{W(y|0)W(y|1)}
$$

It is easy to show that the Bhattacharyya parameter $Z(W)$ is always between $0$ and $1$. Bhattacharyya parameter can be regarded as a measure to determine how good the channel $W$ is. Channels with $Z(W)$ close to zero are almost noiseless, while channels with $Z(W)$ close to one are almost pure-noise channels. More precisely, it can be proved that the probability of error of a binary symmetric memoryless channel (BSM) is upperbounded by its Bhattacharyya parameter. 

The following recursive formulas hold for Bhattacharyya parameters of individual bit-channels in the polar transformation:
\begin{align}
\label{Z_recursion1}
Z(W_{2n}^{(2i-1)}) &\leq 2Z(W_n^{(i)}) - Z(W_n^{(i)})^2\\
\label{Z_recursion2}
Z(W_{2n}^{(2i)}) &= Z(W_n^{(i)})^2
\end{align}
The equality happens in \eq{Z_recursion1} if $W$ is a binary erasure channel. 

The set of \emph{good bit-channels} is defined based on their Bhattacharyya parameters ~\cite{AT,Korada}. Let 
\smash{$[n] \ \raisebox{-0.2ex}{\deff}\ \{1,2,\dots,n\}$}\Strut{2.15ex}{0ex}
and let $\beta \,{<}\, \shalf$ be a fixed positive constant.
Then the index sets of the good bit-channels are given by
\begin{eqnarray}
\label{good-def}
\cG_n(W,\beta)
&\hspace*{-6pt}{\deff}\hspace*{-6pt}&
\left\{\, i \in [n] ~:~ Z(W^{(i)}_n) < 2^{-n^{\beta}}\!\!/n \hspace{1pt}\right\}
\end{eqnarray}

\begin{theorem}
\label{polar_thm1}
~\cite{Arikan,AT} For any BSM channel $W$ and any constant $\beta \,{<}\, \shalf$ we have
$$
\lim_{n \to \infty} \frac{\left|\cG_n(W,\beta)\right|}{n}
\,=\, 
\cC(W)\vspace{1.5ex}
$$
\end{theorem}

\Tref{polar_thm1} readily leads to a construction of capacity-achieving \emph{polar codes}. The idea is to transmit the information bits over the good bit-channels while freezing the input to the other bit-channels to a priori known values, say zeros. The decoder for this constructed code is the successive cancellation decoder of Arikan~\cite{Arikan}. This decoder will be described in more details later in this section. The key property of the encoder-decoder pair of polar codes is summarized in the following theorem. This theorem is (the second part of) 
Proposition\,2 of Arikan~\cite{Arikan}.
\begin{theorem}
\label{polar_thm2}
Let $W$ be a BSM channel and let $k = \left|\cG_n(W,\beta)\right|$.
Suppose that a message $\bU$ is chosen 
uniformly at random from $\{0,1\}^k$,
encoded using the polar encoder, and transmitted
over $W$. Then the probability that the channel
output is not decoded to $\bU$ under successive
cancellation decoding satisfies
$$
P \bigl\{\hU \ne \bU \bigr\} 
\, \leq \,
\sum_{i \in \cG_n(W,\beta)} \!Z(W^{(i)}_n) \leq 2^{-n^{\beta}}.
$$
\end{theorem}

Arikan proposed a low-complex implementation of the successive cancellation (SC) decoder of polar codes \cite{Arikan}. Let $n=2^s$ and suppose that $u^n_1$ be the vector that is multiplied by $\Gs$ and then transmitted over independent copies of $W$. Let $y^n_1$ denote the received word. For $i=1,2,\dots, n$, the decoder computes the likelihood ratio (LR) $L^{(i)}_n$ of $u_i$, given the channel outputs $y^n_1$ and previously decoded $\hat{u}^{i-1}_1$. 
$$
L^{(i)}_n(y^n_1,\hat{u}^{i-1}_1) = \frac{W^{(i)}_n (y^{n}_1,\hat{u}^{i-1}_1 | u_i = 0) }{W^{(i)}_n (y^{n}_1,\hat{u}^{i-1}_1 | u_i = 1)}
$$  
The likelihood functions $L^{(i)}_n$ can be computed recursively as follows. Let $u^j_{1,o}$ and $u^j_{1,e}$ denote the subvectors with odd and even indices, respectively. A straightforward calculation using the bit-channel recursion formulas for $n \geq 1$, gives the following recursive formulas:
\begin{align}
\label{LLR_odd}
&L^{(2i-1)}_n (y^n_1,\hatu^{2i-2}_1)\\
\notag
&= \frac{1+L^{(i)}_{n/2}(y^{n/2}_1,\hatu^{2i-2}_{1,e} \oplus \hatu^{2i-2}_{1,o})
L^{(i)}_{n/2}(y^{n}_{n/2+1},\hatu^{2i-2}_{1,e})}{L^{(i)}_{n/2}(y^{n/2}_1,\hatu^{2i-2}_{1,e} \oplus \hatu^{2i-2}_{1,o})
+L^{(i)}_{n/2}(y^{n}_{n/2+1},\hatu^{2i-2}_{1,e})}\\
\label{LLR_even}
&L^{(2i)}_n(y^n_1,\hatu^{2i-1}_1)\\
\notag
 &= L^{(i)}_{n/2}(y^{n/2}_1,\hatu^{2i-2}_{1,e} \oplus \hatu^{2i-2}_{1,o})^{1-\hatu_{2i-1}}
L^{(i)}_{n/2}(y^n_{n/2+1},\hatu^{2i-2}_{1,e})
\end{align}
If $W^{(i)}_n$ is not a good bit-channel, then the decoder knows that $i$-th bit $u_i$ is set to zero and  therefore, $\hat{u_i} = u_i = 0$. Otherwise, it makes the hard decision based on $L_n^{(i)}$. The total number of LRs that need to be calculated is $n(1+\log n)$. Arikan also proposed an exact order using an $n \times (1+\log n)$ trellis in which the LR calculations are carried out. 

\section{Proposed RS-polar concatenated code}
\label{sec:three}

In this section, we describe our proposed construction for concatenating polar codes with outer block codes. We consider the case when the outer code is a RS code. We establish bounds on the error correction performance and discuss the rate-adaptive construction. 

\subsection{Interleaved concatenated RS-polar codes}

The proposed scheme for concatenating polar codes with outer RS codes is illustrated in Fig. \ref{fig:scheme}.
The symbols generated by a certain number of RS codewords are interleaved and converted into binary streams using a fixed basis to provide the input of the polar encoders. In a specific construction, the bits corresponding to the first symbols of all RS codewords are encoded into one polar codeword. Similarly, the information bits of the second polar codeword constitutes of all bit corresponding to the second symbol coordinates of all outer RS codewords. Hence, polar encoding can be done in parallel, which reduces the encoding latency.

 \begin{figure}
\centering
\includegraphics[width=\linewidth]{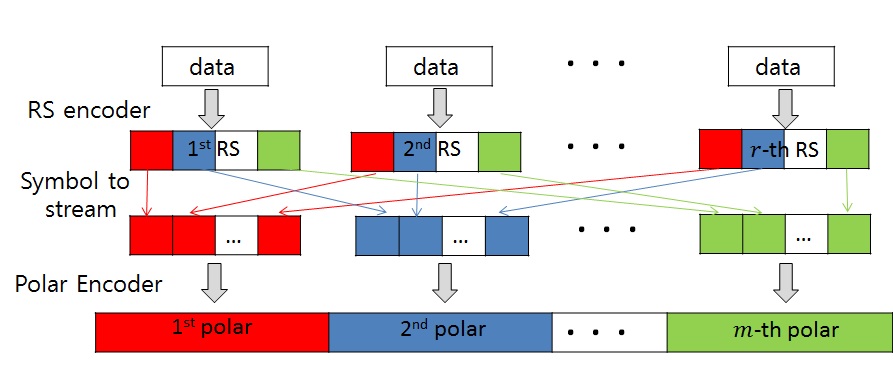}\\
\caption{Proposed concatenation scheme of Polar codes with outer interleaved block codes.
\label{fig:scheme}}
\vspace{-0.5cm}
\end{figure}

More precisely, let $n$ and $m$ denote the lengths of the inner polar code and outer RS code, respectively. Let $k$ denote the number of input bits to each polar encoder i.e. the rate of each inner polar code is $R_I = k/n$. The symbols of the outer RS codes are drawn from the finite field $\F_{2^t}$, with cardinality $2^t$. It is assumed that $k$ is divisible by $t$. Hence, in the proposed scheme, the number of outer RS codes is $r=k/t$ and the number of inner polar codes is $m$. The rates of the outer RS codes will be specified later. Assume that $r$ RS codewords of length $m$ over $\F_{2^t}$ are given. For $i=1,2,\dots,r$, let $(c_{i,1},c_{i,2},\dots,c_{i,m})$ denote the $i$-th codeword.  For $j=1,2,\dots,m$, the $j$-th polar codeword is the output of the polar encoder of rate $k/n$ with the input $\bigl(\mathcal{I}(c_{1,j}),\mathcal{I}(c_{2,j}),\dots,\mathcal{I}(c_{r,j})\bigr)$, where $\mathcal{I}(c)$ maps a symbol $c \in \F_{2^t}$ to its binary image with $t$ bits. Hence, the total length of the concatenated codeword is $N=nm$.
 The interleaver proposed here between the inner and outer codes can be viewed as a structured block interleaver. Other polynomial interleavers may be considered for further improvements. The interleaver plays an important role in this proposed scheme, as it helps to eliminate the need of large field size for the outer RS code in the previous scheme of \cite{BJE}. Since $t=k/r$, $t$ can be fixed as $k$ grows by increasing the number of RS codewords.

\subsection{Rate-adaptive construction of RS-polar concatenated code}
\label{rate-adaptive}
Due to the polarization phenomenon of polar codes, not all bit-channels chosen to carry the information bits have the same reliability. An outer RS code is not actually needed for the strongest bit-channels of the inner polar code, since the corresponding information bits are already well-protected. Our construction guarantees that all symbols of same RS code see the same set of bit-channels, which are different from one RS code to another. In other words, each RS codeword is re-encoded by same bit-channel indices across the different polar codewords. Therefore, the rate of each RS code can be properly assigned to protect the polarized bit-channels in such a way that all the information bits are almost equally protected. Suppose that the probability of error for each of the input bits to the polar code is given. Let $k$ be the information block length of the inner polar code. For $i = 1,2,\dots,k$, let $P_{i}$ be the probability that an error occurs when decoding the $i$-th information bit with the SC decoder, assuming that all the first $i-1$ information bits were successfully recovered. Suppose that the outer RS code is over $\F_{2^t}$, for some integer $t$. Then the total number of RS codes is $k/t$. Then, the first $t$ information bits of each polar codeword form one symbol for the first RS code, the next $t$ information bits form one symbol for the second RS code, etc. If we consider one of the inner polar codes, the probability that the first RS symbol has an error at the output of the SC polar decoder is given by $1-(1-P_1)(1-P_2)\dots(1-P_t)$. In general, for $i=1,2,\dots,k/t$, the probability that a symbol of the $i$-th RS codeword is in error assuming that symbols of preceding RS codes were decoded successfully, is given by $Q_i = 1-(1-P_{it-t+1})(1-P_{it-t+2})\dots(1-P_{it})$.

The design criterion is as follows. Let $\cE$ be the target frame error probability (FEP) of the concatenated code. Then for $i=1,2,..,k/t$, let $\tau_i$ be the smallest positive integer such that
\be{redundancy_assign}
{m \choose \tau_i+1} Q_i^{\tau_i+1} < \frac{t\cE}{k}.
\ee
Then, the proposed rate-adaptive (RA) concatenation scheme deploys a $\tau_i$-error correcting RS code for the $i$-th outer RS code. The following lemma shows that the FEP $\cE$ is guaranteed.
\begin{lemma}
Suppose that the $i$-th outer RS code is a $(m,m-2\tau_i)$ code, for $i=1,2,\dots,k/t$, where $\tau_i$ is determined by \eq{redundancy_assign}. Then, the total frame error probability for the RS-polar concatenated code is less than $\cE$.
\end{lemma}
\begin{proof}
The $i$-th RS decoder is successful if the number of its erroneous symbols is at most $\tau_i$. If some $\tau_i+1$ of symbols out of the total $m$ symbols are in error, then the RS decoder fails. The probability of this event is bounded by
${m \choose \tau_i+1} Q_i^{\tau_i+1}$. This is less than $\frac{t\cE}{k}$ by \eq{redundancy_assign}. Observe that the total FEP of the concatenated code is upper bounded by the union bound on error probability of outer RS codes. Therefore, the probability of frame error is less than $\frac{k}{t} \times \frac{t\cE}{k} = \cE$. 
\end{proof}

The rate-adaptive design criterion, described above, requires knowledge of the individual bit-channel error probability. While it can be precisely calculated for erasure channels, we take a numerical approach to solve this problem for an arbitrary channel, e.g. additive white Gaussian noise (AWGN) channel: Assume that all previous input bits $1,2,\dots,i-1$ are provided to the SC decoder by a genie when the $i$-th bit is decoded. For bit $i = 1,2\dots,n$, the decoder is run for a sufficiently large number of independent inputs to get an estimate of the probability of the event that the $i$-th bit is not successfully decoded, given that the bits indexed by $1,2,\dots,i-1$ were successfully decoded. An alternative way is to use the method introduced in \cite{TV} which provides tight upper and lower bounds on the bit-channel error probability.

\section{Proposed decoding methods for the RS-polar concatenated code}
\label{sec:four}

In this section, we first provide an analysis of asymptotic decoding complexity of the proposed concatenated code. Then several decoding techniques are proposed to improve the finite length performance while the order of decoding complexity is kept the same. 

\subsection{Decoding complexity}

To compute the decoding complexity of the concatenated RS-polar code, we take into account the decoding complexity of both the inner polar code and the outer RS code.

The decoding complexity of the polar code using successive cancellation decoding is given by $O(n\log n)$, with $n$ being the length of the polar code \cite{Arikan}. Since there are $m$ inner polar codes in the proposed concatenated scheme, the total complexity of decoding the inner polar codes is $O(nm\log n)$, which is bounded by $O(N \log N)$.

A well-known hard-decision bounded-distance RS decoding method is the Berlekamp-Massey (BM) algorithm. The BM algorithm is a syndrome-based method which finds the error locations and error magnitudes separately. The decoding complexity is known to be $O(m^2)$ operations over the field $\F_{2^t}$. One main advantage of the proposed concatenated code is that the RS code alphabet size is in the same order of $m$, whereas it is exponential in terms of $n$ in the existing scheme of \cite{BJE}. Gao proposed a syndrome-less RS decoding algorithm  that uses fast Fourier transform and computes the message symbols directly without computing error locations or error magnitudes \cite{G}. For RS codes over arbitrary fields, the asymptotic complexity of syndrome-less decoding based on multiplicative FFT techniques was shown to be $O(m \log^2 m \log\log m)$. Hence by deploying syndrome-less RS decoding, the total complexity of decoding the outer RS codes can be at most $O(nm \log^2 m \log\log m)$ which is bounded by $O(N \log^2 N  \log\log N )$.

Therefore, the total decoding complexity of the proposed concatenated RS-polar code can be asymptotically  bounded by $O(N \log^2 N  \log\log N )$.

\subsection{Successive cancellation RS-polar decoding}

The main drawback of the successive cancellation decoding of polar codes is that once an error occurs, it may propagate through the whole polar codeword. Since the information block of the polar code is protected with an outer RS code, errors in the decoded bits can be corrected using the outer code while the SC decoder evolves. This can potentially mitigate the error propagation problem and consequently results in improvement in the FEP of the proposed scheme at finite block lengths. Hence, we propose the successive cancellation decoding algorithm for our serially concatenated RS-polar code as follows. Let $k$ be the information block length for the inner polar code and $2^t$ be the size of the alphabet for the outer RS code.  Initialize the algorithm with $j=1$. Then
\begin{itemize}
\item
For each of the $m$ polar codewords, decode the $j$-th information sub-block of length $t$, i.e. the information bits indexed by $(j-1)t+1,(j-1)t+2,\dots,jt$. These operations can be done in parallel for the inner polar codewords.
\item
Pass the $mt$ hard-decisioned output bits as $m$ symbols over $\F_{2^t}$ to form the $j$-th RS word, and decode it with the bounded-distance RS decoder.
\item
Update the decoded bits of all $m$ polar codewords using the RS decoder output, and use them to continue SC decoding.
\item
Increase $j$ by one and repeat, while $j < k/t$.
\end{itemize}

Not only the frame error probability of the RS-polar concatenated code can be potentially improved, but also the decoding latency can be significantly reduced, thanks to the parallel processing in the first step of the decoding algorithm explained above.

\subsection{SC Generalized Minimum Distance decoding of RS-polar code}

Generalized minimum distance (GMD) decoding was introduced by Forney in \cite{F}, where the soft information is used with  algebraic bounded-distance decoding to generate a list of codewords. In order to further improve the performance, while keeping the decoding complexity order, we propose SC decoding of the concatenated code with GMD decoding of the outer code. In the concatenated code, the likelihood of each symbol can be computed given the LLRs of the corresponding bits generated by the SC decoder of the inner polar code. The $m$ symbols of each RS word are sorted with respect to their likelihoods.  The $\alpha$ least likely symbols are declared as erasures, where the case of $\alpha=0$ is the same as a regular RS decoding. In conventional GMD decoding, errors and erasures decoding of RS codes is run for $\alpha = \{0,2,4, \dots, d-1\}$, where $d$ is the minimum distance of the RS code. This gives a list of size at most $(d+1)/2$ candidate RS codewords at the output of the decoder. The decoder picks the closest one to the received word, which is then passed to the polar decoders. A naive way of implementing the GMD decoding increases the complexity by a factor of $(d+1)/2$. However, Koetter derived a fast GMD decoding algorithm which removes this factor \cite{K}, and hence GMD can be deployed in decoding our RS-polar code while preserving our decoding complexity bound. Since the SC decoder actually computes the LLR's of the bits in each symbol, the likelihood of each symbol that is passed to RS decoder can be computed. The symbol likelihoods from the different polar SC decoders are used for GMD decoding of each RS code.

\subsection{Near ML SC-GMD decoding }

In the previous subsection, at the last step of GMD RS decoding, the candidate in the generated list of codewords that is the closest one to the received word is picked. Here, we further improve the performance by actually picking the most likely codeword based on soft information from the polar decoder. The first approach is to approximate the symbol probabilities using the bit LLR's generated by the polar decoder. This is not precise, since the bit LLR's in each symbol are not independent. We pick the best codeword from the list generated by GMD decoder based on its estimated probability given by the product of estimated symbol probabilities. We call this approach \emph{SC-GMD-approximate ML or SC-GMD-AML decoding}. In the second approach, the SC decoder of polar code is modified to output the soft information for all the possible symbols. For a symbol constituting of $t$ bits, the LLR of each bit depends on the previous bits in the symbol. In order to compute the exact symbol probabilities, the SC polar decoder computes the probabilities of all the $2^t$ symbols by traversing all the possible $2^t$ paths, for each consecutive $t$ bits. This increases the complexity of polar decoder by a constant factor of $\sum_{i=0}^{t-1}2^i/t$. This enables near ML decoding of outer RS code on top of the SC-GMD decoding of RS-polar code, wherein the GMD decoder picks the most likely candidate from the generated list. We call this method \emph{SC-GMD-ML} decoding. Also, since the SC decoder recursively calculates the LLRs, the LLRs computed along each path are saved so when the correct symbol is picked by the outer decoder, the LLRs computed along the corresponding path are picked to proceed with SC decoding for the next $t$ bits.

\section{Theoretical performance limits}
\label{sec:five}

In this section, upperbounds and lowerbounds on the error correction capability of the proposed concatenated code is derived. Also, it is shown how to set the parameters of the code in order to get the upperbound $2^{-N^{1-\epsilon}}$, for any $\epsilon > 0$, where $N$ is the length of the concatenated code, on the probability of frame error. Furthermore, the performance of the RS-polar concatenated code and Arikan's polar codes is compared assuming transmission over channels with erasure bursts. 

\subsection{Asymptotic analysis of error correction performance}

In our construction, assume all outer RS codes have the same rate $R_o$.
In \Lref{lemma1}, it is shown that the error probability of the concatenated code is bounded by $2^{-(n^{0.5-\eps}(1-R_o)/2 - 1)m}$. Then in \Tref{theorem1}, we prove that $m$, $n$ and $R_o$ can be set in such a way that the error probability of the concatenated code is bounded by $2^{-N^{1-\eps}}$, for any $\eps > 0$, asymptotically, while the concatenated code is still capacity achieving. This significantly improves the error-decay rate compared to polar codes with the same length $N$.

\begin{lemma}
\label{lemma1}
In the proposed RS-polar concatenated scheme, for any $\eps > 0$ and large enough $n$, the probability of frame error is upper bounded by $ 2^{(\frac{-n^{0.5-\eps}(1-R_o)}{2} - 1)m}$, where $n$ and $m$ are the lengths of the inner and outer codes, respectively, and $R_o$ is the rate of outer code.
\end{lemma}
\begin{proof}
Assuming a bounded-distance RS decoder, the error correction capability of RS codes is $\tau = \left\lfloor (1-R_o)m/2\right\rfloor$. The decoder for the whole concatenated scheme fails if more than $\tau$ of the inner polar codewords are in error. Let $P_e$ denote the probability of block error of the inner polar decoder and $\mathcal{E}$ be the frame error probability (FEP) of the concatenated code. $\cE$ can be derived from $P_e$ as follows:
\be{FER}
\cE = \sum^{m}_{i=\tau+1} {m \choose i} P_e^i (1-P_e)^{m-i} \leq {m \choose \tau+1} P_e^{\tau+1}
\ee 
The upper bound on the probability of error holds by the following simple observation. If some $\tau+1$ of the polar codewords are in error, then a decoding error occurs. The bound is derived by counting all the possibilities for the location of $\tau+1$ erroneous symbols. Since some error incidents are counted multiple times, we get an upper bound in \eq{FER}. By plugging in the result of \Tref{polar_thm2} on the probability of error $P_e$ of polar codes into the upper bound in \eq{FER} we get
\begin{equation}
\begin{split}
\label{FER_est2}
\cE &\leq {m \choose \tau+1} 2^{-n^{0.5 - \eps}(\tau+1)}\\
& \leq {m \choose \tau+1} 2^{-n^{0.5 - \eps}m(1-R_o)/2} \\
& < 2^m . 2^{-n^{0.5 - \eps}m(1-R_o)/2} \\
& = 2^{-(n^{0.5 - \eps}(1-R_o)/2 - 1)m}
\end{split}
\end{equation}
\end{proof}
Notice that the bound in \Lref{lemma1} can be tightened using the Stirling's approximation. Instead of simply bounding ${m \choose \tau+1}$ by $2^m$, the more precise Stirling's approximation of $2^{m (H((1-Ro)/2)  - 1/2\log m + O(1))}$, where $H(x)$ is the binary entropy function, can be used. However, this does not affect the result of the next theorem in the asymptotic sense.  
\begin{theorem}
\label{theorem1}
For any $\eps > 0$, the lengths of the inner polar code and outer RS code, and the rate of outer RS code can be set
such that the frame error probability of the concatenated code of total length $N$ is asymptotically upper bounded by $2^{-N^{1-\eps}}$, while the scheme is still capacity-achieving.
\end{theorem}
\begin{proof}
The length of the inner polar code $n$, the length of the outer RS code $m$, and the rate of outer RS code $R_o$
can be set as follows:
\begin{equation} \nonumber
n = N^{\eps},\ m = N^{1-\eps}\ \text{, and}\ R_o = 1 - 4N^{-\eps(0.5-\eps)}.
\end{equation}
Substituting $n$, $m$, and $R_o$ into the bound given by \Lref{lemma1}, one gets
\begin{equation}
\mathcal{E} \leq 2^{-(n^{0.5 - \eps}(1-R_o)/2 - 1)m} = 2^{-N^{1-\eps}}
\end{equation}
as the upper bound on the FEP.
With above settings, $R_o \rightarrow 1$, as $N \rightarrow \infty$. Hence, the
 rate of the concatenated polar code also approaches the capacity, since the inner polar code is proven to be capacity-achieving.
\end{proof}

A lower bound on the probability of error can be also derived by analyzing an optimistic case. The optimistic case is the following: when a polar decoder fails, only one sub-block of the decoded data which contributes to one symbol of the RS outer code is in error. Also, these erroneous symbols are distributed equally over the RS codewords i.e. each RS codeword gets an equal number of erroneous symbols. In this case, the system can support up to $\tau r$ errors in polar decoders where $r$ is the number of outer RS codewords and $\tau$ is the error-correction capability of the outer RS code. In fact $r \approx nR_I / \text{log} \,m$, where $R_I$ is the rate of inner polar code and $\text{log}\,m$ is the number of bits in representation of a symbol in the RS codeword. In this optimistic case, the FEP $\cE_L$ is given by
$$
\cE_L = \sum^{m}_{i=r\tau +1} {m \choose i} P_e^i (1-P_e)^{m-i}
$$
Then plug in the result of \cite{AT}, for large enough $n$, and use the approximation discussed in \Lref{lemma1} along with the Stirling's approximation to get
\begin{equation}
\label{FER_bound2}
\begin{split}
\cE_L &\approx {m \choose r\tau+1} 2^{-n^{0.5 - \eps}(r\tau+1)}\\
& \approx {m \choose r\tau+1} 2^{-n^{0.5 - \eps}nmR_I(1-R_o)/2\text{log} m} \\
& \approx 2^{m H(nR_IR_o/m\log m) -n^{1.5 - \eps}mR_I(1-R_o)/2\text{log} m}\\
\end{split}
\end{equation}

\subsection{Analysis of erasure burst correction performance}

An error burst is a common error pattern in storage systems which consists of a contiguous run of erroneous symbols. An erasure burst of length $d$ is a run of $d$ consecutive symbols that are all erased. 

Suppose that a polar code of length $n=2^s$ is constructed for transmission over on a B-DMC $W$, with respect to the set of good bit-channels $\cG_n(W,\beta)$ for a fixed $\beta \,{<}\, \shalf$. Let $u_1^n$ denote the input message, where $u_i$ carries an information bit for any $i \in \cG_n(W,\beta)$ and otherwise, it is frozen to zero. $u_1^n G_n$ is transmitted through $n$ independent copies of $W$. The received sequence is denoted by $y_1^n$. This is the scenario considered in the following two lemmas. 

\begin{lemma}
\label{erasure_bound}
Let $q \leq s$ and $j \leq 2^{s-q}$ be two positive integers. Assume that an erasure burst of length $2^q$ occurs with starting index $(j-1)2^q+1$ and ending index $j2^q$. Then for any $l$ with $0 \leq l \leq 2^q - 1$, the computed likelihood ratio of $u_{l2^{s-q}+1}$ by successive cancellation decoder is $1$.    
\end{lemma} 
\begin{proof}
Observe that for any $i \leq n/2$, the likelihood ratio $L_n^{(2i-1)}(y_1^n,\hatu_1^{(2i-1)})$ is $1$ if one of the terms $L^{(i)}_{n/2}(y^{n/2}_1,\hatu^{2i-2}_{1,e} \oplus \hatu^{2i-2}_{1,o})$ or 
$L^{(i)}_{n/2}(y^{n}_{n/2+1},\hatu^{2i-2}_{1,e})$ is $1$. This is clear by the recursive formula for LR calculation in \eq{LLR_odd}. Calculation of the LR of $u_{l2^{s-q}+1}$ reduces to calculation of LRs of the form $L^{(l+1)}_{2^q}$, after $s-q$ recursions. One of these terms is of the form $L^{(l+1)}_{2^q}(y_{(j-1)2^q+1}^{j2^q}, ...)$, where the second coordinate is some linear combination of $\hatu_1,\hatu_2,...,\hatu_{l2^{s-q}}$. Since all the output symbols $y_{(j-1)2^q+1}^{j2^q}$ are erased, the later is always $1$. This completes the proof of the lemma. 
\end{proof}
\begin{lemma}
\label{erasure_index}
Let $q \geq s/2$. Then for any $\beta \,{<}\, \shalf$ and large enough $n$, there exists $l \leq 2^q-1$ such that $W^{(l2^{s-q}+1)}_n$ is a good bit-channel. 
\end{lemma}
\begin{proof}
We show that $l = 2^q-1$ satisfies the required condition. Let $Z = Z(W)$. Then by \eq{Z_recursion1} and induction on $r$ it is easy to show that
\be{erasure_index1}
Z(W^{(2^q)}_{2^q}) = Z^{2^q}
\ee
Also, by \eq{Z_recursion2} and induction on $i$ we have
\be{erasure_index2}
Z(W^{((2^q-1)2^{i}+1)}_{2^{i+q}}) \leq 2^i Z(W^{(2^{r})}_{2^q})
\ee
\eq{erasure_index1} and \eq{erasure_index2} together imply that 
$$
Z(W^{((2^q-1)2^{s-q}+1)}_n) \leq 2^{s-q} Z^{2^q} \leq 2^{s/2} Z^{2^{s/2}}
$$
In order to conclude that $(2^q-1)2^{s-q}+1$ is a good bit-channel index, we need to show that for large enough $s$ 
$$
2^{s/2}Z^{2^{s/2}} < 2^{-s}2^{-2^{s\beta}}
$$
Or equivalently, by taking logarithm from both sides
$$
\frac{s}{2}+2^{s/2}\log Z < -s-2^{s\beta}
$$
which holds for large enough $s$ given the fact that $\log Z < 0$ and $\beta \,{<}\, \shalf$. This completes the proof of lemma.  
\end{proof}
\begin{theorem}
\label{erasure_thm}
Consider an Arikan's polar code of length $n$, constructed for transmission over a BSM channel $W$. If an erasure burst of length at least $2\sqrt{n}-1$ occurs, the successive cancellation decoder always fails to recover the transmitted message with probability at least $0.5$.
\end{theorem}
\begin{proof}
Observe that any burst of length $2\sqrt{n}-1 = 2 \times 2^{s/2} - 1$ includes a burst of length $2^{s/2}$ with starting index $(j-1)2^{s/2}+1$ and ending index $j2^{s/2}$, for some $j \leq 2^{s/2}$. Then by \Lref{erasure_index}, there exists a good bit-channel with index of the form $l2^{s/2}+1$. By \Lref{erasure_bound}, the LR of $u_{l2^{s/2}+1}$ is always $1$ and therefore, there is a probability of error $0.5$ associated with this information bit. This completes the proof of theorem.
\end{proof}
The above theorem shows that polar codes are very weak with respect to erasure bursts. In fact, a vanishing fraction of erasures lead to a failure in the successive cancellation decoder with probability $0.5$, while it does not change the effective capacity of channel asymptotically. Next, we show that RS-polar concatenated codes can perform very well in this regard.

Let the total length of the RS-polar concatenated code be $N=nm$, where $n$ is the length of polar code and $m$ is the length of the RS code. Let also $d$ be the minimum distance of the RS code.
\begin{lemma}
\label{erasure_recover}
The RS-Polar concatenated code, with outer code of minimum distance $d$ and inner code of length $n$, can recover from erasure bursts as long as the length of the burst is at most $(d-2)n+1$.
\end{lemma}
\begin{proof}
Observe that a burst of length $(d-2)n+1$ overlaps with at most $d-1$ polar codewords. In the worst case, assume that any polar decoder that encounters at least one erasure declares failure and output erasures to the outer decoder. At most $d-1$ of the inner polar decoders fail in this scenario. The minimum distance of the outer RS code is $d$ which means that it can correct $d-1$ or less erasures. Therefore, RS-polar concatenated code can correct all the erasures resulted from the erasure burst.
\end{proof}
Suppose that $m \leq n$ and the minimum distance of the RS code is at least $4$. Then
$$
(d-2)n+1 \geq 2n+1 \geq 2\sqrt{N}+1
$$
Therefore, a polar code of length $N$ fails to recover from an erasure burst of length $(d-2)n+1$ by \Tref{erasure_thm}, while our RS-polar code can recover from the erasure burst. By fixing $d$ to be any positive integer and letting $n$ to go to infinity, the total probability of error of the concatenated code is $m2^{-n^{\beta}}$, which goes to zero as $m \leq n$. The only constraint on $m$ is that $m \leq n$. Also, it is assumed that $m$ goes to infinity so that the rate of outer RS code approaches $1$. Since the polar code is capacity-achieving, the total concatenated code is also capacity-achieving. All of these happen while a non-concatenated polar code of equal length $N$ fails with probability at least $0.5$ in this scenario, for large enough $N$.

\section{Simulated Numerical Performance \label{sec:fourb}}
\label{sec:six}

\subsection{Simulated performance over AWGN channel}

Transmission over AWGN channel is assumed. The inner polar code of length $2^9 = 512$ and outer Reed-Solomon code of length $15$ over $\F_{2^4}$ are considered. The rates of inner and outer codes are designed such that the total rate of scheme is $1/3$.  We use the RA method explained in Section \ref{rate-adaptive} to construct the concatenated code. The actual probability of error of the bit-channels under SC decoding corresponding to a polar code of length $n=512$ are estimated over an AWGN channel with $E_b/N_0=2$ dB. We take a Monte-Carlo simulation-based approach for this estimation as discussed in Section\,\ref{rate-adaptive}. As an alternative way, one can use the method proposed in \cite{TV}. To design the outer RS code, we follow the criterion proposed for the RA construction in Section\,\ref{rate-adaptive} with the difference that the total rate of the code is fixed rather than the probability of error. To guarantee the total design rate of the concatenated code, the construction is optimized over different possible rates $k/n$ of the inner polar codes as well, and the rate-adaptive construction method is as follows:
\begin{itemize}
\item For $k$ divisible by $4$ and between $170$ (inner rate $1/3$) and $256$ (inner rate $1/2$) do the following.
\item
Pick the best $k$ bit-channels that have smaller probability of errors and sort them with respect to their index.
\item
Set target probability of error $P_e$ for each of the small sub-blocks of length $4$ ($2^4$ is the size of the alphabet for RS code). This replaces $t\cE/k$ in \eq{redundancy_assign}. Then for $i=1,2,..,k/4$, find the $\tau_i$-error correcting RS code according to the criterion proposed in \eq{redundancy_assign}. 
\item
Calculate the total rate of the scheme and compare it with $1/3$. If it is almost $1/3$, then calculate the total FEP and move on to the next $k$. Otherwise, adjust the probability of error $P_e$ accordingly and repeat these steps.
\end{itemize}
At the end, $k = 204$ is picked which results in the lowest frame error rate. 

The performance of the proposed construction with discussed decoding techniques is shown in Figure\,\ref{plot4}. The results are compared with a polar code of the same length $512$. Since all the $15$ inner polar codes in the RS-polar concatenated code are decoded in parallel, the two schemes have the same decoding latency. For the concatenated scheme, the codeword error rate of inner polar codes is defined as the error rate of the inner polar codeword. The aim is to have a fair comparison with a polar code of the same block size $512$ when the rates are equal, where the rate loss due to RS outer code is taken into account. At an block error probability (BLER) of $10^{-4}$, it is observed that the proposed rate-adaptive (RA) concatenated RS-polar code with the proposed SC decoding of the concatenated code has more than $1$ dB SNR gain over the non-concatenated polar code with the same decoding latency. The GMD decoding, on top of SC, helps at higher BLERs, e.g. $0.1$ dB further SNR gain at BLER$=10^{-2}$. Furthermore, there is about $0.2$ dB SNR gain using GMD decoding with approximate ML (GMD-AML) on top of the SC decoding and $0.3$ dB further SNR gain using our proposed SC-GMD-ML decoding of the concatenated code. Also, the proposed SC-GMD-ML decoding algorithm offers more than 2 dB SNR gain over conventional serial decoding of the RS-polar concatenated code, in which the outer RS code is a $(15,11)$ code and the inner polar code is a $(512,232)$ code. The performance of this concatenated scheme is also compared with the outer component Reed-Solomon code with GMD decoding. The $(63,21)$ RS code over $\F_{2^6}$ is picked, where its rate is equal to that of our scheme. The length of its binary representation is $6 \times 63 = 378$ which is close to that of the inner polar code. It is observed that the performance of the RS code with GMD decoding is about $9$ dB SNR worse than that of the proposed scheme with GMD decoding. Therefore, it is not shown in Figure\,\ref{plot4}.   

\begin{figure}[h]
\centering
%

\begin{tikzpicture}

\definecolor{mycolor1}{rgb}{1,0,1}
\definecolor{mycolor2}{rgb}{0,1,1}

\begin{semilogyaxis}[
scale only axis,
width = 2.75in,
height = 2.0625in,
xmin=1, xmax=6,
ymin=1e-06, ymax=1,
xlabel={{\footnotesize $E_b/N_0$ [dB]}},
ylabel={{\footnotesize Block error rate}},
title = {{\footnotesize $n = 512$, rate = $1/3$}},
xmajorgrids,
ymajorgrids,
yminorgrids,
legend entries={{\scriptsize polar code},{\scriptsize regular RS-polar},{\scriptsize RA-SC},{\scriptsize RA-SC-GMD}, {\scriptsize RA-SC-GMD-AML}, {\scriptsize RA-SC-GMD-ML}},
legend style={nodes=right}]

\addplot [
color=blue,
solid,
mark=+,
mark options={solid}
]
coordinates{
 (1,0.3226)
 (1.5,0.1167)
 (2,0.0458)
 (2.5,0.0092)
 (3,0.0018)
 (3.5,0.00029032)
 (4,3.6863e-05)

};

\addplot [
color=red,
solid,
mark=*,
mark options={solid}
]
coordinates{
 (2.5,0.9804)
 (3,0.5701)
 (3.5,0.072)
 (4,0.0029)
 (4.5,0.00011111)

};

\addplot [
color=green,
solid,
mark=triangle,
mark options={solid}
]
coordinates{
 (1,0.8092)
 (1.5,0.2025)
 (2,0.015)
 (2.5,0.0004607)
 (2.7,6.1667e-05)
 (3,4.4444e-06)

};

\addplot [
color=mycolor1,
solid,
mark=x,
mark options={solid}
]
coordinates{
 (1,0.7261)
 (1.5,0.117)
 (2,0.007)
 (2.25,0.0017)
 (2.5,0.00032275)
 (2.7,6.0317e-05)

};

\addplot [
color=mycolor2,
solid,
mark=o,
mark options={solid}
]
coordinates{
 (1,0.5972)
 (1.5,0.0815)
 (1.75,0.0201)
 (2,0.0046)
 (2.25,0.001)
 (2.5,0.00011818)

};

\addplot [
color=black,
solid,
mark=square,
mark options={solid}
]
coordinates{
 (1,0.5138)
 (1.5,0.0355)
 (1.75,0.0073)
 (2,0.0011)
 (2.25,7.7348e-05)

};

\end{semilogyaxis}

\end{tikzpicture}
\caption{Performance of the concatenated scheme using GMD-ML decoding technique}
\label{plot4}
\vspace{-0.5cm}
\end{figure}
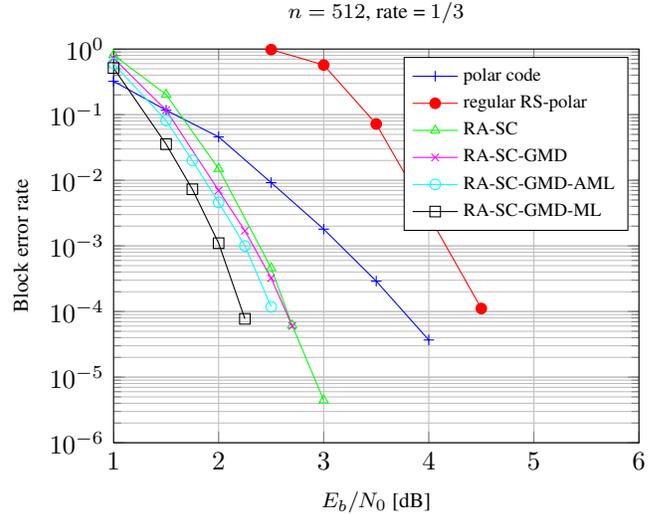

\subsection{Trade-off between polarization order and RS decoding radius} 

It is interesting to analyze the trade-off between the lengths of inner and outer codes assuming that the total length of the concatenated code is fixed. The longer the inner polar code is, the further polarization happens which results in better performance in the inner code. On the other hand, longer outer RS code provides better error correction capability and also a wider range of rates to be chosen by the rate adaptive scheme, thereby making the concatenated code more efficient. 

Recall the proof of \Tref{theorem1} wherein the code parameters are set as follows: the inner block length is $n = N^{\epsilon}$ and the outer block length is $m = N^{1-\epsilon}$. Then the error decay rate is bounded by $2^{-N^{1-\epsilon}}$ as proved in \Tref{theorem1}. Therefore, if $\epsilon$ is decreased, the bound on the error decay rate will be improved, asymptotically. As a result, it is well-justified to say that the smaller the inner polar code is, the further improvement in error decay rate happens in as asymptotic sense. 

In finite block length, there is no analytical way to determine the optimum parameters for inner and outer block lengths. In this subsection, we provide simulation results in an attempt to find out the optimum parameters when the total block length is about $2^{14}$. The total rate of the scheme is fixed and is equal to $1/2$. The rate-adaptive scheme, except for the case of $n=16$, in which there is only one outer RS code, is designed at  $E_b/N_0 = 2$ dB. The FEP varies depending on the scheme and is set for each of them separately in order to maintain the same rate of $1/2$ for the total concatenated code. More details on each design of the concatenated code is as follows: For the case of RS code over $\F_{2^5}$ of length $31$ and inner polar code of length $512$ (simply denoted by RS($31$)-polar($512$)), the optimum value of $k$, the information block length of polar code, is 295. There are $295/5=59$ outer RS codes in total. For the case of RS($63$)-polar($256$), $k = 144$ and there are $24$ outer RS codes of length $63$ over $\F_{2^6}$. Consequently, for the concatenated code RS($127$)-polar($128$), $k=70$ and there are $10$ outer RS codes of length $127$ over $\F_{2^7}$. We have also chosen the extreme case of RS($1023$)-polar($16$) to compare with other schemes, where $k=10$ and there is only one RS code of length $1023$ over $\F_{2^{10}}$. The simulation results are shown in Figure\,\ref{plot2}. It can be observed that the concatenated code RS($31$)-polar($512$) shows a better performance comparing to the other schemes. Notice that the extreme case of RS($1023$)-polar($16$) has a much sharper slope of frame error rate in terms of SNR comparing to other cases. However, its performance is much worse which may be due to the relatively very small number of polarization levels. 

\begin{figure}[h]
\centering
%

\begin{tikzpicture}

\definecolor{mycolor1}{rgb}{1,0,1}

\begin{semilogyaxis}[
scale only axis,
width = 2.75in,
height = 2.0625in,
xmin=1, xmax=6,
ymin=0.0001, ymax=1,
xlabel={$E_b/N_0$},
ylabel={frame error rate},
xmajorgrids,
ymajorgrids,
yminorgrids,
legend entries={{\scriptsize RS(15)-polar(1024)}, {\scriptsize RS(31)-polar(512)}, {\scriptsize RS(63)-polar(256)}, {\scriptsize RS(127)-polar(128)}, {\scriptsize RS(1024)-polar(16)}},
legend style={nodes=right}]

\addplot [
color=blue,
solid,
mark=triangle,
mark options={solid}
]
coordinates{
 (1,0.9615)
 (1.5,0.1887)
 (2,0.0054)
 (2.25,0.0012)

};

\addplot [
color=red,
solid,
mark=x,
mark options={solid}
]
coordinates{
 (1,0.9709)
 (1.5,0.1235)
 (2,0.0026)
 (2.25,0.00056769)

};

\addplot [
color=black,
solid,
mark=o,
mark options={solid}
]
coordinates{
 (1,1)
 (1.5,0.4425)
 (2,0.0061)
 (2.25,0.000555)

};

\addplot [
color=green,
solid,
mark=square,
mark options={solid}
]
coordinates{
 (1,1)
 (1.5,0.7463)
 (2,0.093)
 (2.5,0.011)
 (3,0.003)

};

\addplot [
color=mycolor1,
solid,
mark=+,
mark options={solid}
]
coordinates{
 (3.2,0.9901)
 (3.3,0.3849)
 (3.4,0.1756)
 (3.6,0.001)
 (3.7,0.000129)

};

\end{semilogyaxis}

\end{tikzpicture}
\caption{Trade-offs between the lengths of inner and outer codes}
\label{plot2}
\vspace{-0.5cm}
\end{figure}
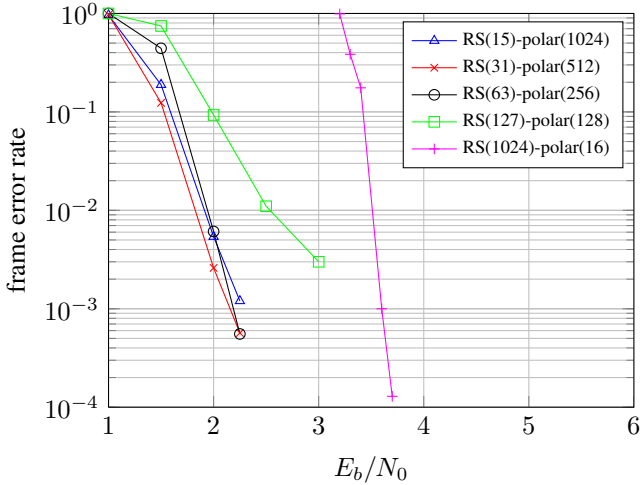

\subsection{Performance gains from concatenation over channels with erasure bursts}

We have also observed the advantage of our RS-polar concatenated code comparing to non-concatenated polar codes over channels with erasure bursts through simulations. A Gilbert-Elliot model for producing erasure bursts is considered. This model is based on a Markov chain with two states for the channel. The channel is either in the good state, denoted by G, or in the bad state, denoted by B. When the channel is in state G, stays in this state with probability $P$ and changes to state B with probability $1-P$. Likewise, when the channel is in state B, stays in this state with probability $Q$ and changes to state G with probability $1-Q$. 

For simulation, state G is considered as binary erasure channel with probability of erasure $0.1$ and state B is always an erasure. In the RS-polar concatenated code, the length of inner polar code is set as $2^9=512$. The outer RS code is a $(15,11)$ code which can correct up to four erasures. The total length of the concatenated code is $N = 512*15 = 7680$. The rate of inner polar code is set as $0.68$ so that the total rate of the concatenated code is $0.5$. The transition probabilities in the Gilbert-Elliot model are picked as $P=0.9999$ and $Q=0.99$. This way, the average length of a run of bad channel states is $1/(1-Q) = 100$, which is close to $\sqrt{N}$. Likewise, the average length of consecutive good channel states is $1/(1-P)=10000$, which is close to $N$. This model resembles the conditions in \Tref{erasure_thm}.    
Intuitively, the inner polar codes corrects the erasures resulted from the the BEC$(0.1)$ in good state while the outer RS code corrects the failures resulted from erasure bursts. The probability of frame error estimated as $6 \times 10^{-4}$ using Monte Carlo simulations. On the other hand, the non-concatenated polar code of length $2^{13} = 8192$ does not perform well in this scenario, as expected. The probability of frame error is estimated as $0.1$ using Monte Carlo simulations. Notice that in this scenario, the computed likelihood ratios are either $0$, $1$ or $\infty$ when the symbols are transmitted through binary erasure channel. If the computed likelihood ratio is either $0$ or $\infty$, it means that the corresponding input bit is recovered with no error. Otherwise, the corresponding input bit is erased. Therefore, the successive cancellation decoder can not be improved using list decoding algorithm proposed in \cite{TV2}.

\section{Discussions and conclusion}
\label{sec:seven}

In this paper, we proved that by carefully concatenating the recently invented polar codes with Reed-Solomon codes, a significant improvement in the error-decay rate compared to non-concatenated polar codes is possible which ultimately fills the gap with the information theoretic bound. The parameters of the scheme can be set to inherit the capacity-achieving property of polar codes while working in the same regime of low complexity. The proposed concatenated code is shown to perform well over channels with erasure bursts, while it is proved that the original polar codes are very weak in this regard. We developed several construction methods and decoding techniques to improve the performance at finite block lengths, which is a step to making polar codes more practical. There are some directions for future work as follows. The methods described in this paper can be in general applied to concatenation of polar codes with non-binary block codes. The construction and decoding methods can also be used when the outer code is a binary code by grouping each $t$ bits, where $t$ is a complexity parameter to be optimized. If the outer code has a low complexity soft decoding algorithm, then the decoding techniques based on GMD decoding of RS codes can be extended to this case as well.

\end{document}